\documentclass[conference]{IEEEtran}
\IEEEoverridecommandlockouts
\usepackage[noadjust]{cite}
\usepackage{amsmath,amssymb,amsfonts}
\usepackage{algorithmic}
\usepackage{graphicx}
\usepackage{textcomp}
\usepackage{xcolor}
\usepackage{comment}
\usepackage{mathptmx}
\usepackage{amsthm}
\usepackage{nicefrac}
\usepackage{dsfont}
\newtheorem{thm}{Theorem}
\newtheorem{lem}{Lemma}

\theoremstyle{definition}
\newtheorem{defn}{Definition}

\newtheorem{rem}{Remark}
\def\BibTeX{{\rm B\kern-.05em{\sc i\kern-.025em b}\kern-.08em
    T\kern-.1667em\lower.7ex\hbox{E}\kern-.125emX}}

\usepackage{bbm}

\usepackage{etoolbox}    
\usepackage{dirtytalk}
\usepackage{changepage}
\usepackage{enumitem}

\newenvironment{sketchproofachievable}{%
   \proof}{\endproof}
\newenvironment{proofconverse}{%
  \proof}{\endproof}

\definecolor{NYUviolet}{HTML}{57068c} 	
\definecolor{NYUlight}{HTML}{8900e1} 	
\definecolor{NYUdark}{HTML}{330662} 	
\definecolor{NYUnight}{HTML}{220337} 	

\begin{document}

\title{Database Matching Under Adversarial Column Deletions\\
\thanks{A shorter version of this paper was published in Proc. IEEE Information Theory Worksop (ITW), Saint-Malo, France, April 2023. This work is supported by National Science Foundation grants 1815821 and 2148293.}}

\author{Serhat Bakirtas, Elza Erkip\\
 NYU Tandon School of Engineering\\
Emails: \{serhat.bakirtas, elza\}@nyu.edu}

\maketitle

\begin{abstract}
The de-anonymization of users from anonymized microdata through matching or aligning with publicly-available correlated databases has been of scientific interest recently. While most of the rigorous analyses of database matching have focused on random-distortion models, the adversarial-distortion models have been wanting in the relevant literature. In this work, motivated by synchronization errors in the sampling of time-indexed microdata, matching (alignment) of random databases under adversarial column deletions is investigated. It is assumed that a constrained adversary, which observes the anonymized database, can delete up to a $\delta$ fraction of the columns (attributes) to hinder matching and preserve privacy. Column histograms of the two databases are utilized as permutation-invariant features to detect the column deletion pattern chosen by the adversary. The detection of the column deletion pattern is then followed by an exact row (user) matching scheme. The worst-case analysis of this two-phase scheme yields a sufficient condition for the successful matching of the two databases, under the near-perfect recovery condition. A more detailed investigation of the error probability leads to a tight necessary condition on the database growth rate, and in turn, to a single-letter characterization of the adversarial matching capacity. This adversarial matching capacity is shown to be significantly lower than the \say{random} matching capacity, where the column deletions occur randomly. Overall, our results analytically demonstrate the privacy-wise advantages of adversarial mechanisms over random ones during the publication of anonymized time-indexed data.
\end{abstract}

\section{Introduction}
\label{sec:introduction}
With the ever-increasing popularity of smartphones, IoT devices, and big data applications, the user data gathered by companies and institutions has been growing as well. This user-level microdata is then published or shared for scientific and/or commercial purposes, after \emph{anonymization} which refers to the removal of any explicit identifiers. However, concerns over the insufficiency of simple anonymization have been articulated by the scientific~\cite{ohm2009broken} and corporate~\cite{bigdata} communities. These concerns were further validated and amplified as researchers devised practical privacy attacks on real data~\cite{naini2015you,datta2012provable,narayanan2008robust,sweeney1997weaving,takbiri2018matching} to show the vulnerability of anonymization on its own.

In the light of the above practical privacy attacks on databases, several groups initiated rigorous analyses of the database matching problem which has applications beyond privacy, such as image processing~\cite{sanfeliu2002graph}, computer vision~\cite{galstyan2021optimal}, single-cell biological data alignment~\cite{zhu2021robust,tran2020benchmark} and DNA sequencing, which is shown to be equivalent to matching bipartite graphs~\cite{blazewicz2002dna}. Matching of correlated databases has also been rigorously investigated from information-theoretic and statistical perspectives~\cite{cullina,shirani8849392,dai2019database,kunisky2022strong,tamir2022joint,zeynepdetecting2022,bakirtas2021database,bakirtas2022matching,bakirtas2022seeded}. In \cite{cullina}, Cullina \emph{et al.} derived sufficient conditions for successful matching and a converse result using perfect recovery as the error criterion. In~\cite{shirani8849392}, Shirani \emph{et al.} considered a pair of anonymized and obfuscated databases and derived necessary and sufficient conditions on the \emph{database growth rate} for reliable matching, in the presence of noise on the database entries, under near-exact recovery criterion. In~\cite{dai2019database,kunisky2022strong,tamir2022joint}, the matching of a pair of databases with jointly-Gaussian attributes is considered. In~\cite{tamir2022joint,zeynepdetecting2022}, the necessary and the sufficient conditions for detecting whether two Gaussian databases are correlated are investigated.
\begin{figure}[t]
\centerline{\includegraphics[width=0.5\textwidth,trim={0 15cm 2cm 0},clip]{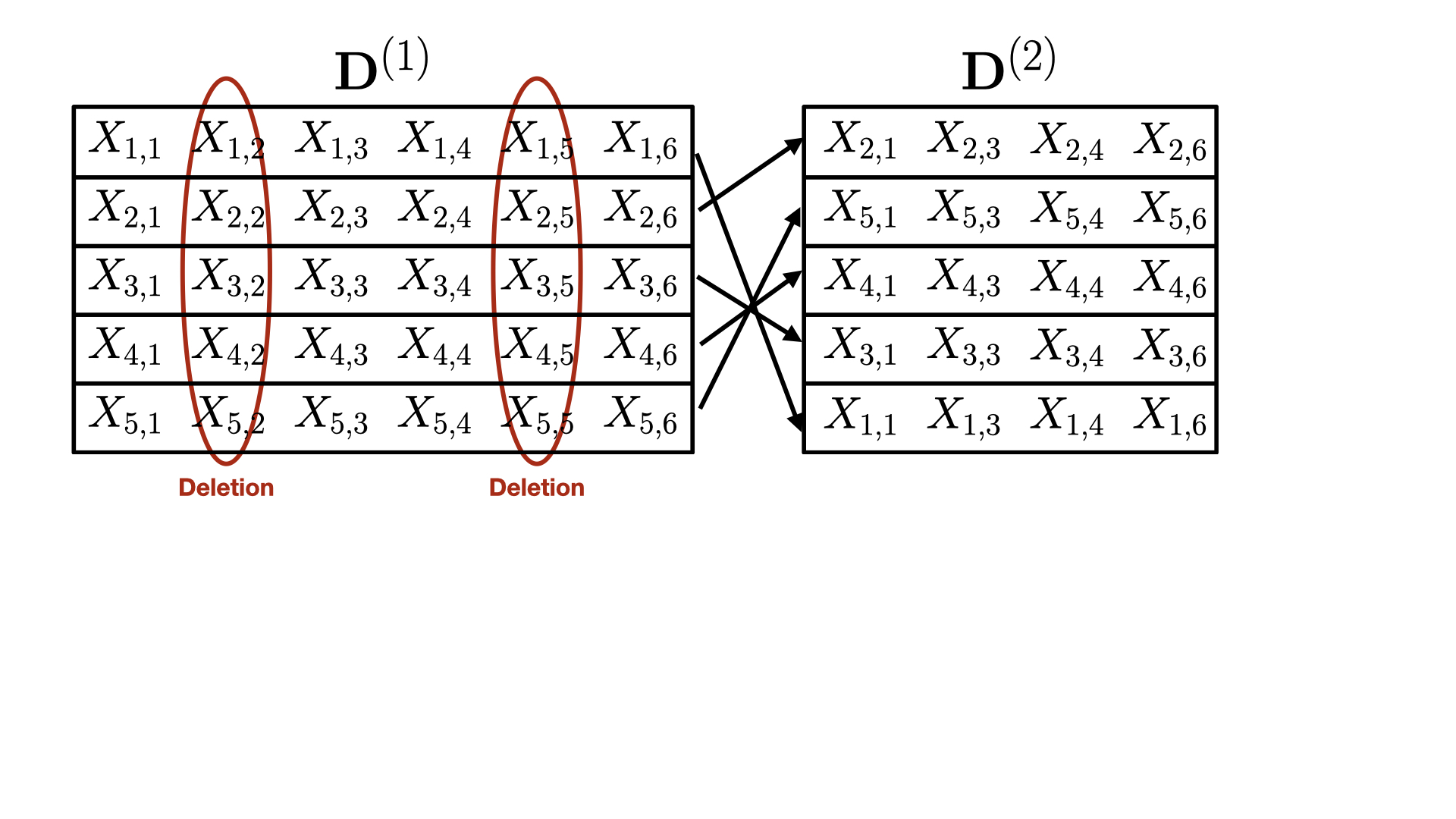}}
\caption{An illustrative example of database matching under column deletions. The columns circled in red are deleted. Our goal is to estimate the row permutation {${\Theta}_n$} which is in this example given as; {$\boldsymbol{\Theta}_n(1)=5$}, {$\boldsymbol{\Theta}_n(2)=1$}, {$\boldsymbol{\Theta}_n(3)=4$},
{$\boldsymbol{\Theta}_n(4)=3$} and {$\boldsymbol{\Theta}_n(5)=2$}, by matching the rows of $\mathbf{D}^{(1)}$ and $\mathbf{D}^{(2)}$, under column deletions with $I_{\text{del}}=(2,5)$. Here the $i$\textsuperscript{th} row of $\mathbf{D}^{(1)}$ corresponds to the {$\Theta_n(i)$\textsuperscript{th}} row of $\mathbf{D}^{(2)}$.}
\label{fig:intro}
\end{figure}

In~\cite{bakirtas2021database,bakirtas2022matching,bakirtas2022seeded}, motivated by the synchronization errors in the sampling of time-series datasets, we investigated the matching of two databases of the same number of users (rows), but with different numbers of attributes (columns). In our model, one of the databases suffers from \emph{random column repetitions}. Under this model, we devised various algorithms to detect the underlying repetition pattern. In~\cite{bakirtas2022seeded}, we showed that in the noisy setting, a batch of seeds whose size $B_n$ grows logarithmic in the number of rows $m_n$ of the database, can be utilized for the detection of deletion locations and replicas can be detected without any seeds. Similarly, in~\cite{bakirtas2022matching}, we showed in the noiseless setting, the repetition detection can be performed without any seeds through a repetition detection algorithm. These repetition detection algorithms were then followed by joint-typicality-based matching schemes which allowed us to derive achievable database growth rates. Then, we proved tight converse results, characterizing the matching capacities of the database matching problem under noiseless and noisy random column repetitions. 

Motivated by potential settings in which a privacy-preserving mechanism denies the sampling of the most informative attributes after observing the anonymized database, our objective in this paper is to investigate the necessary and sufficient conditions for the successful matching of database rows under adversarial column deletions. Unlike the previous work~\cite{cullina,shirani8849392,dai2019database,bakirtas2021database,bakirtas2022matching,bakirtas2022seeded,zeynepdetecting2022,kunisky2022strong,tamir2022joint,chen2022one} where distortions, in the form of noise and/or synchronization errors, are random, we assume a constrained-adversarial model as in channel coding literature~\cite{2627,8849568,kash2011zero,5205859,bassily2014causal}. We assume that synchronization errors, in the form of column deletions, are chosen by the constrained adversary where the constraint is of the form of a fractional column deletion budget. An example of these column deletions is illustrated in Figure~\ref{fig:intro}. We stress that this \say{adversary} here is in fact a privacy-preserving entity whose goal is to hinder matching of the databases. Under this assumption, we improve upon and utilize the histogram-based detection algorithm of~\cite{bakirtas2022matching} and then propose an exact sequence matching algorithm. We note that this adversarial model forces us to focus on the worst-case scenario and in turn, prohibits the use of typicality and Fano's inequality, as done in~\cite{shirani8849392,bakirtas2021database,bakirtas2022matching,bakirtas2022seeded}. Therefore, the Hamming distances between the rows (users) of the databases become crucial in our analyses, as is often the case in the adversarial channel literature~\cite{bassily2014causal}.

The organization of this paper is as follows: We formulate the problem in Section~\ref{sec:problemformulation}. We state our main result on the adversarial matching capacity and prove its achievability part in Section~\ref{sec:mainresult}. Next, we prove the converse part in Section~\ref{sec:converse}. Finally, in Section~\ref{sec:conclusion} the results and ongoing work are discussed.

\noindent{\em Notation:} $[n]$ denotes the set of integers $\{1,...,n\}$. We denote matrices with uppercase bold letters and for a matrix $\mathbf{D}$, its $(i,j)$\textsuperscript{th} entry with $D_{i,j}$. Furthermore, by $A^n$, we denote a row vector consisting of scalars $A_1,\dots,A_n$ and the indicator of event $E$ by $\mathds{1}_E$. $H$ denotes Shannon's entropy~\cite[Chapter 2]{cover2006elements}. The logarithms, unless stated explicitly, are in base $2$. 

\section{Problem Formulation}
\label{sec:problemformulation}

Throughout this work, we utilize the following definitions, some of which are similar to~\cite{shirani8849392,bakirtas2021database,bakirtas2022matching,bakirtas2022seeded}, to formulate our database matching problem. 

\begin{defn}{\textbf{(Unlabeled Database)}}\label{defn:unlabeleddb}
An ${(m_n,n,p_{X})}$ \emph{unlabeled database} is a randomly generated ${m_n\times n}$ matrix ${\mathbf{D}=\{D_{i,j}\in\mathfrak{X}\}}$ with \emph{i.i.d.} entries drawn according to the distribution $p_X$ with a finite discrete support $\mathfrak{X}=\{1,\dots,|\mathfrak{X}|\}$.
\end{defn}

\begin{defn}{\textbf{(Adversary, Column Deletion Pattern)}}\label{defn:deletionpattern}
The \emph{column deletion pattern} $I_{\text{del}}=\{i_1,i_2,...,i_d\}\subseteq [n]$ is a vector consisting of $d$ entries, chosen by the \say{adversary} after observing the unlabeled database $\mathbf{D}$. The parameter $\delta\triangleq {\nicefrac{d}{n}}$ is called the \emph{deletion budget}.
\end{defn}

Different from~\cite{bakirtas2022matching,bakirtas2022seeded} where column repetitions (deletions and replications) are considered, in this work, we focus on a deletion-only setting. This is because the additional replicas either have no effect on the matching performance as in the noiseless case~\cite{bakirtas2022matching} or offer additional information acting as a repetition code of random length in the noisy setting and in turn, boost the matching performance \cite{bakirtas2022seeded}. Hence, it is expected for any privacy mechanism that tries to hinder the matching process not to allow the replication of entries. Therefore in the adversarial repetition setting, it is natural to focus on the deletion-only case. 

Note that the column deletion pattern $I_{\text{del}}$, as described in Definition~\ref{defn:deletionpattern}, is not independent of the unlabeled database $\mathbf{D}$, as assumed in~\cite{bakirtas2021database,bakirtas2022seeded,bakirtas2022matching}. We further assume that deletions occur columnwise, \emph{i.e.,} every row experiences the same column deletion pattern. Here, $I_{\text{del}}$ indicates which columns of $\mathbf{D}$ are deleted. When $j\in I_{\text{del}}$, the $j$\textsuperscript{th} column of $\mathbf{D}$ is said to be \emph{deleted}. Otherwise, it is said to be \emph{retained}.

\begin{figure}[t]
\centerline{\includegraphics[width=0.5\textwidth]{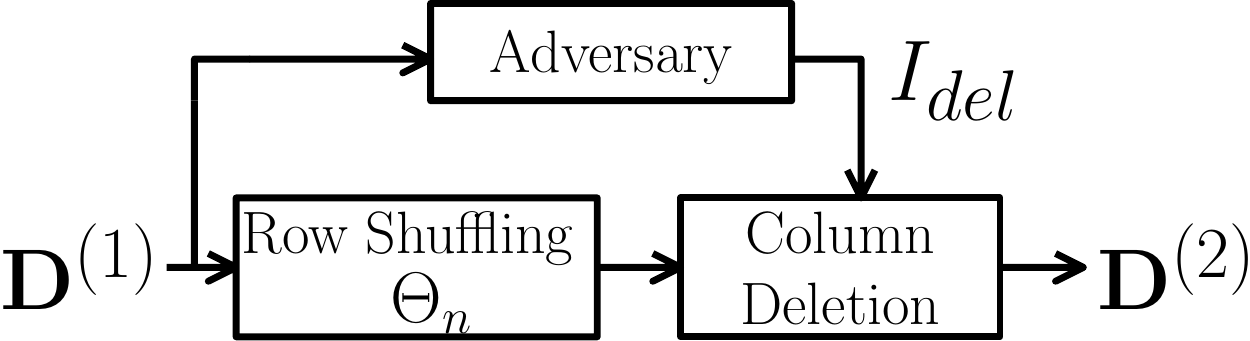}}
\caption{Relation between the unlabeled database $\mathbf{D}^{(1)}$ and the column deleted labeled one, $\mathbf{D}^{(2)}$.}
\label{fig:dmc}
\end{figure}

\begin{defn}{\textbf{(Column Deleted Labeled Database)}}\label{defn:labeleddb}
Let $\mathbf{D}^{(1)}$ be an ${(m_n,n,p_{X})}$ unlabeled database. Let $I_{\text{del}}=(i_1,\dots,i_d)$ be a column deletion pattern, $\boldsymbol{\Theta}_n$ be a uniform permutation of $[m_n]$, independent of $(\mathbf{D}^{(1)},I_{\text{del}})$. Given $\mathbf{D}^{(1)}$ and $I_{\text{del}}$, $\mathbf{D}^{(2)}$ is called the \emph{column deleted labeled database} if the respective $(i,j)$\textsuperscript{th} entries ${D}^{(1)}_{i,j}$ and ${D}^{(2)}_{i,j}$ of $\mathbf{D}^{(1)}$ and $\mathbf{D}^{(2)}$ have the following relation:
\begin{align}
D^{(2)}_{i,j}&=
    \begin{cases}
      E , &  \text{if } j\in I_{\text{del}}\\
      D^{(1)}_{\boldsymbol{\Theta}_n^{-1}(i),j} & \text{if } j\notin I_{\text{del}}
    \end{cases} 
\end{align}
where ${{D}^{(2)}_{i,j}=E}$ corresponds to ${D}^{(2)}_{i,j}$ being the empty string.

The $i$\textsuperscript{th} row of $\mathbf{D}^{(2)}$ is said to correspond to the {$\boldsymbol{\Theta}_n^{-1}(i)$\textsuperscript{th}} row of $\mathbf{D}^{(1)}$, where $\boldsymbol{\Theta}_n$ is called the \emph{labeling function}.
\end{defn}

The relationship between $\mathbf{D}^{(1)}$ and $\mathbf{D}^{(2)}$, as described in Definition~\ref{defn:labeleddb}, is illustrated in Figure~\ref{fig:dmc}. Our main goal is to estimate the labeling function $\Theta_n$ with $\mathbf{D}^{(1)}$ and $\mathbf{D}^{(2)}$ without observing $I_{\text{del}}$. In other words, the deletion locations are unknown.

In this work, we assume that there is no noise on the retained entries after row shuffling and column deletions, as is often done in the synchronization channel literature~\cite{cheraghchi2020overview}. 

Note that in this setting, although the deletions are not random, the matching error event is still random due to the random natures of $\mathbf{D}^{(1)}$ and $\boldsymbol{\Theta}_n$. Furthermore, since the deletion indices are chosen in an adversarial fashion, we adopt a worst-case near-exact recovery performance metric in the following definition:
\begin{defn}{\textbf{(Successful Matching Scheme)}}
A \emph{matching scheme} is a sequence of mappings {${\phi_n: (\mathbf{D}^{(1)},\mathbf{D}^{(2)})\mapsto \hat{\boldsymbol{\Theta}}_n }$} where $\mathbf{D}^{(1)}$ is the unlabeled database, $\mathbf{D}^{(2)}$ is the column deleted labeled database and $\hat{\boldsymbol{\Theta}}_n$ is the estimate of the correct labeling function $\boldsymbol{\Theta}_n$. The scheme $\phi_n$ is said to be \emph{successful} against an adversary with a $\delta$-deletion budget, if 
\begin{align}
    \Pr(\forall I_{\text{del}}=(i_1,\dots,i_{n\delta})\subseteq [n], \hat{\Theta}_n(J)\neq \Theta_n(J)) &\overset{n\to\infty}{\longrightarrow} 0   \label{eq:proberror}
\end{align}
where the index $J$ is drawn uniformly from $[m_n]$ and the dependence of the matching scheme $\hat{\Theta}_n$ on the column deletion index set $I_{\text{del}}$ is omitted for brevity.
\end{defn}

We stress that both in database matching and correlation detection settings, the relationship between the row size $m_n$, the column size $n$ and the database distribution parameters are the parameters of interest~\cite{kunisky2022strong,zeynepdetecting2022,tamir2022joint}. Note that as the row size $m_n$ increases for fixed column size $n$, matching becomes harder. This is because for a given column size $n$, as the row size $m_n$ increases, so does the probability of mismatch as a result of having a larger candidate row set. Furthermore, as stated in~\cite[Theorem 1.2]{kunisky2022strong}, for distributions with parameters constant in $n$ and $m_n$, the regime of interest is the logarithmic regime where $n\sim \log m_n$. Thus, we utilize the \emph{database growth rate} introduced in~\cite{shirani8849392} to characterize the relationship between the row size $m_n$ and the column size $n$.  

\begin{defn}{\textbf{(Database Growth Rate)}}
The \emph{database growth rate} $R$ of an ${(m_n,n,p_X)}$ unlabeled database is defined as 
\begin{align}
    R&=\lim\limits_{n\to\infty} \frac{1}{n}\log m_n.
\end{align}
\end{defn}

\begin{defn}{\textbf{(Achievable Database Growth Rate)}}\label{defn:achievable}
Consider a sequence of ${(m_n,n,p_X)}$ unlabeled databases, an adversary with a $\delta$-deletion budget and the resulting sequence of column deleted labeled databases. A database growth rate $R$ is said to be \emph{achievable} if there exists a successful matching scheme when the unlabeled database has growth rate $R$.
\end{defn}

\begin{defn}{\textbf{(Adversarial Matching Capacity)}}\label{defn:matchingcapacityadversarial}
The \emph{adversarial matching capacity} $C^{\text{adv}}(\delta)$ is the supremum of the set of all achievable rates corresponding to a database distribution $p_X$ and an adversary with a $\delta$-\emph{deletion budget}.
\end{defn}

In this paper, our main goal is to characterize the adversarial matching capacity $C^{\text{adv}}(\delta)$, by proposing matching schemes and a tight upper bound on all achievable database growth rates. Since we are interested in the supremum of achievable rates, throughout this work, we will assume a positive database growth rate, \emph{i.e.,} $R>0$.

\section{Main Result and Achievability}
\label{sec:mainresult}
In this section, we present our main result on the adversarial matching capacity (Theorem~\ref{thm:mainresult}). We prove the achievability part of Theorem~\ref{thm:mainresult} in this section and the converse part in Section~\ref{sec:converse}.

\begin{thm}{\textbf{(Adversarial Matching Capacity)}}\label{thm:mainresult}
Consider a database distribution $p_X$ and an adversary with a $\delta$-\emph{deletion budget}. Then, the adversarial matching capacity is
\begin{align}
    C^{\text{adv}}(\delta) &= \begin{cases}
    D(\delta\|1-\hat{q}),&\text{if } \delta\le 1-\hat{q}\\
    0, &\text{if } \delta> 1-\hat{q}
    \end{cases}
\end{align}
where $\hat{q} \triangleq \sum_{x\in\mathfrak{X}} p_X(x)^2$ and $D(.\|.)$ denotes the Kullback-Leibler divergence~\cite[Chapter 2.3]{cover2006elements} between two Bernoulli distributions with given parameters.
\end{thm}
Before proceeding with the proof of Theorem~\ref{thm:mainresult}, we first compare the matching capacities under adversarial column deletions and under random column deletions, as characterized in~\cite{bakirtas2022matching}.

Note that using~\cite[Theorem 1]{bakirtas2022matching}, we can argue that when each column is deleted independently with probability $\delta$, independent of the unlabeled database $\mathbf{D}^{(1)}$, the \say{random} matching capacity becomes
\begin{align}
    C^\text{random}(\delta) = (1-\delta) H(X).
\end{align}

The matching capacities for random and adversarial deletions as a function of the deletion probability/budget are illustrated in Figure~\ref{fig:adversarialcapacity}. For $\delta\le 1-\hat{q}$, the matching capacity is significantly reduced when the column deletions are adversarial rather than random. Furthermore for $\delta>1-\hat{q}$, the $C^{\text{adv}}(\delta)=0$ whereas ${C^\text{random}(\delta) = (1-\delta) H(X)>0}$, suggesting that for a deletion budget/probability $\delta>1-\hat{q}$, successful matching with a positive database growth rate is possible only when the deletions are random.

\begin{figure}[t]
\centerline{\includegraphics[width=0.45\textwidth]{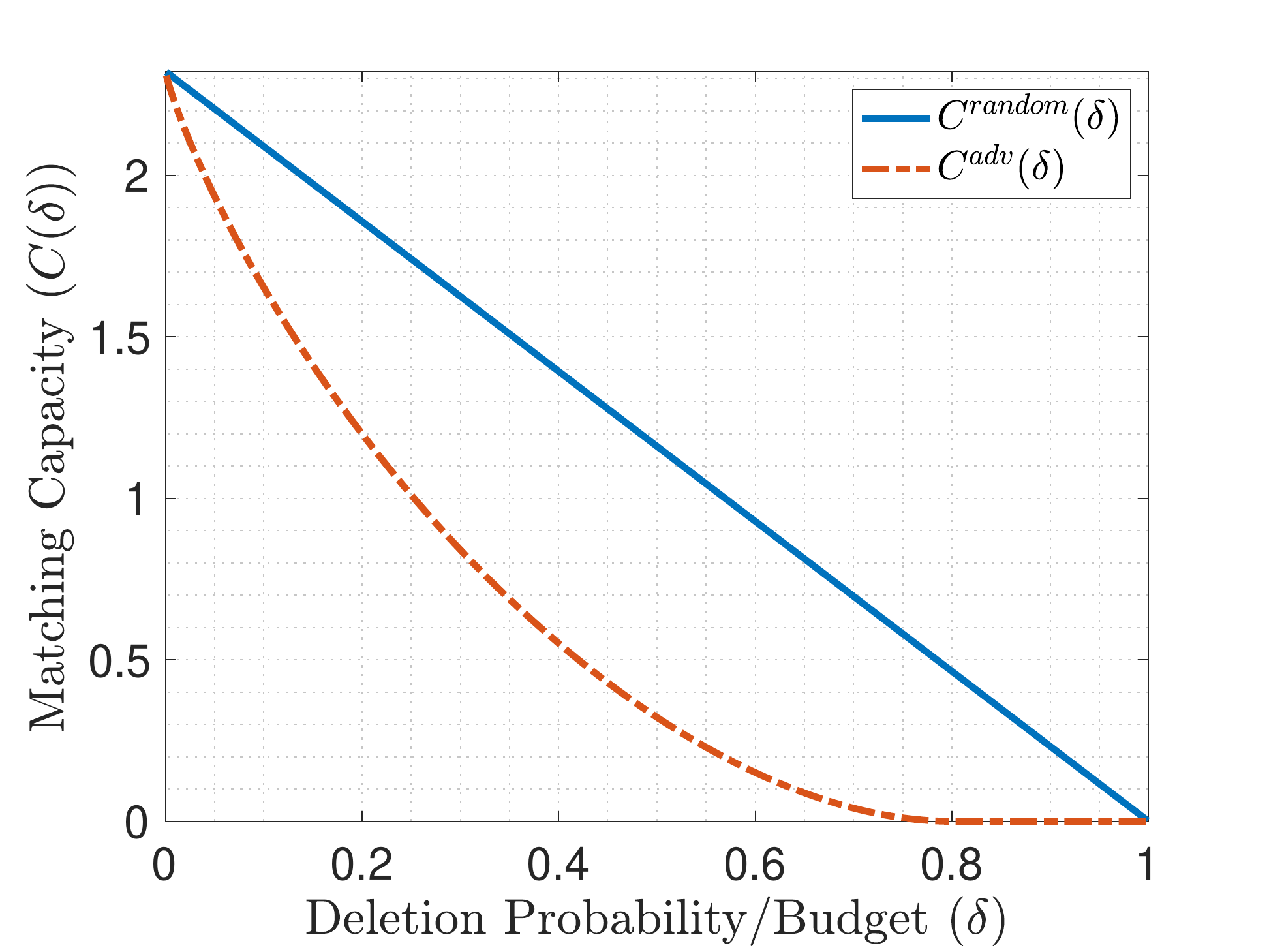}}
\caption{Matching capacities $C$ vs. deletion probability/budget ($\delta$) when $X\sim \text{Unif}(\mathfrak{X})$, $\mathfrak{X}=[5]$. Notice that in this case $\hat{q}=0.2$ and for $\delta>1-\hat{q}=0.8$ the adversarial matching capacity $C^{\text{adv}}(\delta)$ is zero, while the random matching capacity $C^\text{random}(\delta)$ is positive.}
\label{fig:adversarialcapacity}
\end{figure}

The rest of this section is on the proof of the achievability part of Theorem~\ref{thm:mainresult}. In Section~\ref{subsec:deletiondetection}, we discuss our \emph{histogram-based} deletion detection algorithm which is a modified version of the one used in~\cite{bakirtas2022matching} and prove a stronger asymptotic performance than in~\cite{bakirtas2022matching}. Then, in Section~\ref{subsec:matchingscheme}, we prove the achievability of Theorem~\ref{thm:mainresult} through the utilization of the histogram-based detection algorithm and exact sequence matching.

\subsection{Histogram-Based Deletion Detection}\label{subsec:deletiondetection}
We propose to detect the deletions by extracting permutation-invariant features of the columns of $\mathbf{D}^{(1)}$ and $\mathbf{D}^{(2)}$, similar to~\cite{bakirtas2022matching,bakirtas2022seeded}. Our histogram-based deletion detection algorithm works as follows: First, we construct the histogram matrices $\mathbf{H}^{(1)}$ and $\mathbf{H}^{(2)}$ where the $j$\textsuperscript{th} column $H^{(r)}_j$ of $\mathbf{H}^{(r)}$ denotes the histogram of the $j$\textsuperscript{th} column of $\mathbf{D}^{(r)}$, $r=1,2$. More formally, for $r=1,2$ we have
\begin{align}
    {H}^{(r)}_{i,j} &\triangleq \sum\limits_{t=1}^{m_n} \mathbbm{1}_{\left[{D}^{(r)}_{t,j}= i \right]},\forall j
\end{align}
where $K_n$ denotes the column size of $\mathbf{D}^{(2)}$.

Next, we find the estimate $\hat{I}_{\text{del}}$ the column deletion pattern $I_{\text{del}}$ as follows: We start with the initialization $\hat{I}_{\text{del}}=\varnothing$. Then for all $j\in[n]$, if the $j$\textsuperscript{th} column $H^{(1)}_j$ of $\mathbf{H}^{(1)}$ is absent in $\mathbf{H}^{(2)}$, we announce the $j$\textsuperscript{th} column of $\mathbf{D}^{(1)}$ to be deleted, assigning $\hat{I}_{\text{del}}\leftarrow\hat{I}_{\text{del}}\cup {j}$. Otherwise, we infer that the $j$\textsuperscript{th} column of $\mathbf{D}^{(1)}$ is retained.

Observe that the only possibility of an error in the procedure above is when $H^{(1)}_{i}=H^{(1)}_{j}$ for some $i,j\in[n]$ with $i\in I_{\text{del}}$ and $j\notin I_{\text{del}}$. Therefore as long as $\smash{H^{(1)}_{j}}$ are unique, our deletion detection algorithm is error-free.

In the following lemma, we derive a sufficient condition on the relationship between $m_n$ and $n$ for the asymptotic uniqueness of the column histograms.
\begin{lem}{\textbf{(Asymptotic Uniqueness of the Histograms)}}\label{lem:histogram}
Let $H^{(1)}_j$ denote the histogram of the $j$\textsuperscript{th} column of $\mathbf{D}^{(1)}$.
Then,
\begin{align}
    \Pr\left(\exists i,j\in [n],\: i\neq j,H^{(1)}_i={H}^{(1)}_j\right)\to 0 \text{ as }n\to \infty
\end{align}
if $m_n=\omega(n^\frac{4}{|\mathfrak{X}|-1})$.
\end{lem}
\begin{proof}
See Appendix~\ref{proof:histogram}.
\end{proof}

\begin{rem}
    Observe that the order relation derived in Lemma~\ref{lem:histogram} ($m_n=\omega(n^{\nicefrac{4}{|\mathfrak{X}|-1}})$) is better than the one derived in \cite[Lemma 1]{bakirtas2022matching} ($m_n=\omega(n^{4})$), where histograms are \say{collapsed} for tractability in the Markov case. Although the weaker order relation of \cite{bakirtas2022matching} is still satisfied for any positive database growth rate $R>0$, the novel stronger result would be of interest in the zero-rate regime, where $m_n$ is not necessarily exponential in $n$.
\end{rem}

\subsection{Row Matching Scheme and Achievability}\label{subsec:matchingscheme}
We are now ready to prove the achievability part of Theorem~\ref{thm:mainresult}. 
\vspace*{-0.75em}
\begin{sketchproofachievable}
We focus on $\delta\le 1-\hat{q}$ first. For a given pair of matching rows, WLOG, $X_1^n$ of $\mathbf{D}^{(1)}$ and $Y_l^{K_n}$ of $\mathbf{D}^{(2)}$ with $\boldsymbol{\Theta}_n(1)=l$, let $P_e\triangleq \Pr(\hat{\boldsymbol{\Theta}}_n(1)\neq l)$ be the probability of error of the following matching scheme:
\begin{enumerate}[label=\textbf{\arabic*)},leftmargin=1.3\parindent]
\item Construct the histogram vectors ${{H}}_i^{(1)}$ and ${{H}}_j^{(2)}$ as described above, where $K_n=n(1-\delta)$ denotes the column size of $\mathbf{D}^{(2)}$.
\item Check the uniqueness of the columns ${H}^{(1)}_j$ $j\in[n]$ of $\mathbf{H}^{(1)}$. If there are at least two which are identical, declare a \emph{detection error} whose probability is denoted by $\mu_n$. Otherwise, proceed with Step~3.
\item Construct the estimated column deletion pattern $\hat{I}_{\text{del}}$ as described above. Note that conditioned on Step~2, this step is error-free.
\item Obtain $\Tilde{\mathbf{D}}^{(1)}$ from $\mathbf{D}^{(1)}$ by discarding the columns whose indices lie in $\hat{I}_{\text{del}}$. Note that at this step $\Tilde{\mathbf{D}}^{(1)}$ and $\mathbf{D}^{(2)}$ have the same size. 
\item Match the $l$\textsuperscript{th} row $Y^{K_n}_{l}$ of $\mathbf{D}^{(2)}$ with the $1$\textsuperscript{st} row $X^n_1$ of $\mathbf{D}^{(1)}$, assigning $\hat{\boldsymbol{\Theta}}_n(1)=l$ if the $1$\textsuperscript{st} row $\tilde{X}_1^{K_n}$ of $\tilde{\mathbf{D}}^{(1)}$ is the only row of $\tilde{\mathbf{D}}^{(1)}$ equal to $Y^{K_n}_{l}$. Otherwise, declare a \emph{collision error}.
\end{enumerate}
Let $I(\delta)$ be the set of all deletion patterns with $n\delta$ deletions. For the matching rows $X^n_1$, $Y^k_l$ of $\mathbf{D}^{(1)}$ and $\mathbf{D}^{(2)}$, define the pairwise adversarial collision probability between $X^n_1$ and $X^n_i$ for any $i\in[m_n]\setminus\{1\}$ as
\begin{align}
    P_{\text{col,i}}&\triangleq \Pr(\exists \hat{I}_{\text{del}} \in I(\delta):\: {X}_i([n]\setminus\hat{I}_{\text{del}})=Y_l^{K_n})\\
    &=\Pr(\exists \hat{I}_{\text{del}} \in I(\delta):\: {X}_i([n]\setminus\hat{I}_{\text{del}})={X}_1([n]\setminus\hat{I}_{\text{del}})).
\end{align}
where ${X}_i([n]\setminus\hat{I}_{\text{del}})$ is the vector obtained from $X_i^n$ by discarding the elements whose indices lie in $\hat{I}_{\text{del}}$.

Note that the event $\exists \hat{I}_{\text{del}} \in I(\delta):{X}_i([n]\setminus\hat{I}_{\text{del}})={X}_1([n]\setminus\hat{I}_{\text{del}})$ is equivalent to the case when the Hamming distance between $X^n_i$ and $X^n_1$ being upper bounded by $n\delta$. In other words,
\begin{align}
    P_{\text{col,i}} &= \Pr(d_H(X_1^n,X_i^n)\le n\delta) \label{eq:collision}
\end{align}
where $d_H$ denotes the Hamming distance. More formally,
\begin{align}
    d_H(X_1^n,X_i^n) &= \sum\limits_{j=1}^n \mathbbm{1}_{[X_{1,j}\neq X_{i,j}]}
\end{align}
Due to the \emph{i.i.d.} nature of the database elements, $d_H(X_1^n,X_i^n)\sim \text{Binom}(n,1-\hat{q})$, where $\hat{q}=\sum_{x\in\mathfrak{X}} p_X(x)^2$. Thus, for any $\delta\le 1-\hat{q}$, using Chernoff bound~\cite[Lemma 4.7.2]{ash2012information}, we have
\begin{align}
    P_{\text{col,i}} &= \Pr(d_H(X_1^n,X_i^n)\le n\delta)\\
    &\le 2^{-n D(\delta\|1-\hat{q})}\label{eq:advchernoff}
\end{align}

Thus, given the correct labeling for $Y^k_l\in\mathbf{D}^{(2)}$ is $X^n_1\in\mathbf{D}^{(1)}$, the probability of error $P_e$ can be bounded as
\begin{align}
    P_e  &\le \Pr(\exists i\in[m_n]\setminus\{1\}: \tilde{X}_i^{K_n}=\tilde{X}_1^{K_n})\\
    &\le \sum\limits_{i=2}^{2^{n R}} P_{col,i}+ \mu_n\\
    &\le  2^{n R} P_{\text{col,2}}+\mu_n\label{eq:Perowwiseadv}
\end{align}
where \eqref{eq:Perowwiseadv} follows from the fact the the rows are \emph{i.i.d.} and thus $P_{\text{col,i}}=P_{\text{col,2}},\:\forall i\in[m_n]\setminus\{1\}$. Combining \eqref{eq:advchernoff}-\eqref{eq:Perowwiseadv}, we get
\begin{align}
    P_e &\le  2^{n R} \Pr(d_H(X_1^n,X_i^n)\le n\delta) + \mu_n\\
    &\le 2^{n R} 2^{-n D(\delta\|1-\hat{q})}+ \mu_n\\
    &= 2^{-n \left[D(\delta\|1-\hat{q})-R\right]}+\mu_n
\end{align}
By Lemma~\ref{lem:histogram}, $\mu_n\to0$ as $n\to \infty$. Thus, we argue that any rate $R$ satisfying 
\begin{align}
    R&<D(\delta\|1-\hat{q})
\end{align}
is achievable. The rest of the proof trivially follows from the non-negativity of achievable database growth rate for any ${\delta\ge 1-\hat{q}}$.
\end{sketchproofachievable}

We stress that the use of a rowwise matching scheme after the deletion detection phase instead of matching at the database level does not cause a performance loss in terms of achieving the adversarial matching capacity, as we prove in Section~\ref{sec:converse}.
\section{Converse}
\label{sec:converse}
In this section, we show that the achievable rate derived in Section~\ref{sec:mainresult} is in fact tight, by proving a tight upper bound on the all achievable database growth rates and in turn on the adversarial matching capacity $C^{\text{adv}}(\delta)$.
\begin{proofconverse}
Let $R$ be the database growth rate, $\delta$ be the deletion budget of the adversary and $P_e$ be the probability that the scheme is unsuccessful for a uniformly-selected row, WLOG $X_1^n$. In other words, let $P_e\triangleq \Pr(\hat{\boldsymbol{\Theta}}_n(1)\neq \boldsymbol{\Theta}_n(1))\to 0$ as $n\to\infty$. Then, recalling~\eqref{eq:collision}, we have
\begin{align}
    P_e &= \Pr(\exists i\in[m_n]\setminus\{1\}: d_H(X_1^n,X_i^n)\le n\delta)\\
    &= 1 - \Pr(\forall i\in[m_n]\setminus\{1\}: d_H(X_1^n,X_i^n)> n\delta)\label{eq:advconverse1}\\
    &= 1 -\prod\limits_{i=2}^{m_n} \Pr(d_H(X_1^n,X_i^n)> n\delta)\\
    &=  1 -\prod\limits_{i=2}^{m_n} [1-\Pr(d_H(X_1^n,X_i^n)\le n\delta)]\\
    &= 1-[1-\Pr(d_H(X_1^n,X_2^n)\le n\delta)]^{m_n-1} \label{eq:advconverse2}
\end{align}
where \eqref{eq:advconverse1}-\eqref{eq:advconverse2} follow from the fact that the rows of $\mathbf{D}^{(1)}$ are \emph{i.i.d.} Since $D_{n,2}\sim\text{Binom}(n,1-\hat{q})$, for ${\delta\le 1-\hat{q}}$, from~\cite[Lemma 4.7.2]{ash2012information}, we obtain
\begin{align}
    \Pr(D_{n,2}\le n\delta) &\ge \frac{2^{-n D(\delta\|1-\hat{q})}}{\sqrt{2n}}\label{eq:advChernoffLB}
\end{align}
Plugging \eqref{eq:advChernoffLB} into \eqref{eq:advconverse2}, we get
\begin{align}
    P_e&\ge 1-\left[1-\frac{2^{-n D(\delta\|1-\hat{q})}}{\sqrt{2n}}\right]^{m_n-1}
\end{align}

Now let $y=-\frac{2^{-n D(\delta\|1-\hat{q})}}{\sqrt{2n}}\in(-1,0)$. Then, we get
\begin{align}
    P_e&\ge 1-(1+y)^{m_n-1}
\end{align}
Since $y\ge -1$, and $m_n\in\mathbb{N}$, we have
\begin{align}
    1+y(m_n-1)&\le (1+y)^{m_n-1}\le e^{y (m_n-1)}\label{eq:bernoulli}
\end{align}
where the LHS of \eqref{eq:bernoulli} follows from Bernoulli's inequality~\cite[Theorem 1]{brannan2006first} and the RHS of \eqref{eq:bernoulli} follows from the fact that
\begin{align}
    \forall x\in\mathbb{R},\hspace{1em} \forall r\in\mathbb{R}_{\ge0} \hspace{1em} (1+x)^r &\le e^{x r}
\end{align}
Thus, we get
\begin{align}
    P_e&\ge 1-(1+y)^{m_n-1}\\
    &\ge 1-e^{y (m_n-1)}\\
    &\ge 0
\end{align}
since $y<0$, $m_n-1>0$. Note that since $P_e\to 0$, by the Squeeze Theorem~\cite[Theorem 2]{brannan2006first}, we have
\begin{align}
    \lim\limits_{n\to\infty} 1-e^{y (m_n-1)}&\to 0
\end{align}
This, in turn, implies $y m_n\to0$ since the exponential function is continuous everywhere. In other words,
\begin{align}
    \lim\limits_{n\to\infty}& -\frac{2^{-n D(\delta\|1-\hat{q})}}{\sqrt{2n}} m_n 
    \to0
\end{align}
Equivalently, from the continuity of the logarithm function, we get
\begin{align}
    \lim\limits_{n\to\infty}&-n D(\delta\|1-\hat{q})+\log m_n - \frac{1}{2}\log (2 n)\to - \infty\\
    \lim\limits_{n\to\infty}&-n \left[D(\delta\|1-\hat{q})-\frac{1}{n}\log m_n+\frac{\log (2 n)}{2n}\right]\to-\infty\\
    \lim\limits_{n\to\infty}& \left[D(\delta\|1-\hat{q})-\frac{1}{n}\log m_n+\frac{\log (2 n)}{2n}\right]\ge 0
\end{align}
This implies
\begin{align}
        D(\delta\|1-\hat{q})&\ge \lim\limits_{n\to\infty} \frac{1}{n}\log m_n\\
        &= R
\end{align}
finishing the proof for $\delta\le 1-\hat{q}$. Thus, combining with the achievability result of Section~\ref{subsec:matchingscheme}, we have showed that
\begin{align}
    C^{\text{adv}}(\delta)&=D(\delta\|1-\hat{q})
\end{align}
for $\delta\le 1-\hat{q}$. 

We argue that for $\delta>1-\hat{q}$, the adversarial matching capacity is zero, by using two facts: \emph{i)} Since any increase in the adversarial deletion budget hinders matching, the adversarial matching capacity satisfies
\begin{align}
    C^{\text{adv}}(\delta)&\le C^{\text{adv}}(\delta^\prime),\hspace{1em}\forall \delta^\prime\le \delta
\end{align}
and \emph{ii)} $C^{\text{adv}}(1-\hat{q})=0$. Thus, $\forall \delta>1-\hat{q}$, $C^{\text{adv}}(\delta)=0$. This finishes the proof.
\end{proofconverse}

\section{Conclusion}
\label{sec:conclusion}
In this work, we have investigated the database matching problem under adversarial column deletions. We have showed that, similar to the random repetitions setting, column histograms could be used to detect the column deletion pattern. Then, we proposed an exact sequence matching algorithm and derived an achievable database growth rate. Finally, we proved that this achievable database growth rate is in fact tight and thus obtained a complete single-letter characterization of the adversarial matching capacity. Comparing adversarial and random matching capacities, we showed that the adversarial matching capacity is significantly lower than the random matching capacity. Furthermore, we observed that when the deletion probability/budget exceeds a threshold, which is based on the database distribution, the adversarial matching capacity becomes zero, while the random matching capacity is strictly positive. Overall, our results show that adopting an adversarial privacy mechanism, instead of random sampling, can hinder the matching of two correlated databases, providing insight into privacy-preserving publication of user microdata.

\typeout{}
\bibliographystyle{IEEEtran}
\bibliography{references}
    
\appendix
\subsection{Proof of Lemma~\ref{lem:histogram}}\label{proof:histogram}
For brevity, we let 
\begin{align}
    \mu_n\triangleq \Pr(\exists i,j\in [n],\: i\neq j,H^{(1)}_i=H^{(1)}_j).
\end{align}
Notice that since the entries of $\mathbf{D}^{(1)}$ are \emph{i.i.d.}, $H^{(1)}_i$ are \emph{i.i.d.} Multinomial$(m_n,p_X)$ random variables. Then,
\begin{align}
    \mu_n&\le n^2 \Pr(H^{(1)}_1=H^{(1)}_2)\\
        &=n^2 \sum\limits_{h^{|\mathfrak{X}|}} \Pr(H^{(1)}_1=h^{|\mathfrak{X}|})^2
\end{align}
where the sum is over all vectors of length $|\mathfrak{X}|$, summing up to $m_n$. Let $m_i\triangleq h(i)$, $\forall i\in\mathfrak{X}$. Then,
\begin{align}
    \Pr(H^{(1)}_1=h^{|\mathfrak{X}|})&= \binom{m_n}{m_1,m_2,\dots,m_{|\mathfrak{X}|}} \prod\limits_{i=1}^{|\mathfrak{X}|} p_X(i)^{m_i}
\end{align}
Hence, we have
\begin{align}
    \mu_n &\le n^2\sum\limits_{m_1+\dots+m_{|\mathfrak{X}|}=m_n} \binom{m_n}{m_1,m_2,\dots,m_{|\mathfrak{X}|}}^2 \prod\limits_{i=1}^{|\mathfrak{X}|} p_X(i)^{2 m_i}\label{eq:multinomial}
\end{align}
where $\smash{\binom{m_n}{m_1,m_2,\dots,m_{|\mathfrak{X}|}}}$ is the multinomial coefficient corresponding to the $|\mathfrak{X}|$-tuple $(m_1,\dots,m_{|\mathfrak{X}|})$ and the summation is over all possible non-negative indices $m_1,\dots,m_{|\mathfrak{X}|}$ which add up to $m_n$.

From~\cite[Theorem 11.1.2]{cover2006elements}, we have
\begin{align}
    \prod\limits_{i=1}^{|\mathfrak{X}|} p_X(i)^{2 m_i}=2^{-2m_n(H(\Tilde{p})+D(\Tilde{p}\|p_X))}\label{eq:covertype}
\end{align}
where $\Tilde{p}$ is the type corresponding to $|\mathfrak{X}|$-tuple ${(m_1,\dots,m_{|\mathfrak{X}|})}$: 
\begin{align}
    \Tilde{p}&=\left(\frac{m_1}{m_n},\dots,\frac{m_{|\mathfrak{X}|}}{m_n}\right)
\end{align}
From Stirling's approximation~\cite[Chapter 3.2]{cormen2022introduction}, we get
\begin{align}
    \binom{m_n}{m_1,m_2,\dots,m_{|\mathfrak{X}|}}^2\le \frac{e^2} {(2\pi)^{|\mathfrak{X}|}} m_n^{1-|\mathfrak{X}|} \Pi_{\Tilde{p}}^{-1} 2^{2m_n H(\Tilde{p})}\label{eq:stirling}
\end{align}
where $\Pi_{\Tilde{p}}=\prod_{i=1}^{|\mathfrak{X}|} \Tilde{p}(i)$.

Combining \eqref{eq:multinomial}-\eqref{eq:stirling}, we get
\begin{align}
    \mu_n\le \frac{e^2} {(2\pi)^{|\mathfrak{X}|}} n^2 m_n^{1-|\mathfrak{X}|} \sum\limits_{\Tilde{p}} \Pi_{\Tilde{p}}^{-1} 2^{-2m_n D(\Tilde{p}\|p_X)}
\end{align}

Let \begin{align}
    T=\sum\limits_{\Tilde{p}} \Pi_{\Tilde{p}}^{-1} 2^{-2m_n D(\Tilde{p}\|p_X)} = T_1 + T_2
\end{align} where
\begin{align}
    T_1&=\sum\limits_{\Tilde{p}:D(\Tilde{p}\|p_X)> \frac{\epsilon_n^2}{2\log_e 2}} \Pi_{\Tilde{p}}^{-1} 2^{-2m_n D(\Tilde{p}\|p_X)}\label{eq:T1iid}\\
    T_2&=\sum\limits_{\Tilde{p}:D(\Tilde{p}\|p_X)\le\frac{\epsilon_n^2}{2\log_e 2}} \Pi_{\Tilde{p}}^{-1} 2^{-2m_n D(\Tilde{p}\|p_X)}.\label{eq:T2iid}
\end{align}
Here, $\epsilon_n$, which is described below in more detail, is a small positive number decaying with $n$.

First, we look at $T_2$. From Pinsker's inequality~\cite[Lemma 11.6.1]{cover2006elements}, we have
\begin{align}
    D(\Tilde{p}\|p_X)\le \frac{\epsilon_n^2}{2\log_e 2}\Rightarrow \text{TV}(\Tilde{p},p_X)\le \epsilon_n
\end{align}
where TV denotes the total variation distance. Therefore
\begin{align}
    \left|\{\Tilde{p}:D(\Tilde{p}\|p_X)\le \frac{\epsilon_n^2}{2\log_e}\}\right|&\le |\{\Tilde{p}:\text{TV}(\Tilde{p},p_X)\le \epsilon_n\}|\notag \\
    &= O(m_n^{|\mathfrak{X}|-1}\epsilon_n^{|\mathfrak{X}|-1})
\end{align}
where the last equality follows from the fact in a type we have $|\mathfrak{X}|-1$ degrees of freedom, since the sum of the $|\mathfrak{X}|$-tuple $(m_1,\dots,m_{|\mathfrak{X}|})$ is fixed. 
Furthermore, when $\text{TV}(\Tilde{p},p_X)\le \epsilon_n$, we have 
\begin{align}
    \Pi_{\Tilde{p}} &\ge \prod\limits_{i=1}^{|\mathfrak{X}|} (p_X(i)-\epsilon_n)\ge \Pi_{p_X}-\epsilon_n \sum\limits_{i=1}^{|\mathfrak{X}|} \prod\limits_{j\neq i} p_X(j)
\end{align}
Hence
\begin{align}
    \Pi_{\Tilde{p}}^{-1}&\le \frac{1}{\Pi_{p_X}-\epsilon_n \sum\limits_{i=1}^{|\mathfrak{X}|} \prod\limits_{j\neq i} p_X(j)}
\end{align}
and 
\begin{align}
    T_2 &\le \frac{1}{\Pi_{p_X}-\epsilon_n \sum\limits_{i=1}^{|\mathfrak{X}|} \prod\limits_{j\neq i} p_X(j)} O(m_n^{|\mathfrak{X}|-1}\epsilon_n ^{|\mathfrak{X}|-1})\\
    &= O(m_n^{|\mathfrak{X}|-1}\epsilon_n^{|\mathfrak{X}|-1})
\end{align}
for small $\epsilon_n$.
    
Now, we look at $T_1$. Note that since $m_i\in \mathbb{Z}_+$, we have ${\Pi_{\Tilde{p}}\le m_n^{|\mathfrak{X}|}}$, suggesting the multiplicative term in the summation in~\eqref{eq:T1iid} is polynomial with $m_n$. If $m_i=0$ we can simply discard it and return to Stirling's approximation with the reduced number of categories. Furthermore, from~\cite[Theorem 11.1.1]{cover2006elements}, we have
\begin{align}
    \left|\{\Tilde{p}:D(\Tilde{p}\|p_X)> \frac{\epsilon_n^2}{2\log_e 2}\}\right|&\le |\{\Tilde{p}\}|\\&\le (m_n+1)^{|\mathfrak{X}|}
\end{align}
suggesting the number of terms which we take the summation over in~\eqref{eq:T1iid} is polynomial with $m_n$ as well. Therefore, as long as ${m_n \epsilon_n^2\to\infty}$, $T_1$ has a polynomial number of elements which decay exponentially with $m_n$. Thus
\begin{align}
    T_1\to0\text{ as }n\to\infty\label{eq:t1iid}
\end{align}
    
Define \begin{align}
    U_i&=e^2 (2\pi)^{-|\mathfrak{X}|} m_n^{1-|\mathfrak{X}|} T_i,\quad i=1,2\label{eq:ui}
\end{align} and choose ${\epsilon_n=m_n^{-\frac{1}{2}} V_n}$ for some $V_n$ satisfying ${V_n=\omega(1)}$ and ${V_n=o(m_n^{1/2})}$. Thus, $U_1$ vanishes exponentially fast since ${m_n\epsilon_n^2=V_n^2\to\infty}$ and \begin{align}
    U_2&=O(\epsilon_n^{|\mathfrak{X}|-1})=O(m_n^{(1-|\mathfrak{X}|)/2} V_n^{(|\mathfrak{X}|-1)}).\label{eq:u2}
\end{align}
Combining \eqref{eq:t1iid}-\eqref{eq:u2}, we have 
\begin{align}
    U=U_1+U_2=O(m_n^{(1-|\mathfrak{X}|)/2} V_n^{(|\mathfrak{X}|-1)})
\end{align}
and we get 
\begin{align}
    \mu_n\le n^2 O(m_n^{(1-|\mathfrak{X}|)/2} V_n^{(|\mathfrak{X}|-1)})
\end{align}
By the assumption ${m=\omega(n^\frac{4}{|\mathfrak{X}|-1})}$, we have ${m_n=n^\frac{4}{|\mathfrak{X}|-1} Z_n}$ for some $Z_n$ satisfying  ${\lim\limits_{n\to\infty} Z_n=\infty}$. Now, taking ${V_n=o(Z_n^{1/2})}$ (e.g.~$V_n=Z_n^{1/3}$), we get
\begin{align}
    \mu_n&\le O(n^2 n^{-2} Z_n^{(1-|\mathfrak{X}|)/2} V_n^{(|\mathfrak{X}|-1)})
    = o(1)
\end{align}
Thus $m=\omega(n^\frac{4}{|\mathfrak{X}|-1})$ is enough to have $\mu_n\to0$ as $n\to\infty$.\qed
\end{document}